\documentclass[11pt,onecolumn]{article}

\usepackage{bagrow}

\usepackage{amssymb}
\usepackage{amsthm}
\usepackage{hyperref} 
\usepackage[usenames,dvipsnames,svgnames,table]{xcolor}

\def\factor{0.5} 

\def\suppmat{Appendix}
\def\SM{App.\ A}

\def\Gp{G^{\prime}}
\def\Djs{D_\mathrm{JS}}
\def\ER{Erd{\H o}s-R{\'e}nyi}
\def\BA{Barab{\'a}si-Albert}
\newcommand\avg[1]{\ensuremath{\left< #1 \right>}}
\def\celegans{\emph{C.\ elegans}}
\newtheorem{theorem}{Theorem}[section]
\newtheorem{definition}{Definition}[section]

\newcommand\KL[2]{\ensuremath{\mathit{KL}\left( #1 \,||\, #2 \right)}} 

\usepackage{graphicx}

\usepackage{enumitem}
\setlist[enumerate]{leftmargin=.5in}
\setlist[itemize]{leftmargin=.5in}

\begin{document}

\title{An information-theoretic, all-scales approach to comparing networks}

\date{July 25, 2019}

\author[1,2,*]{James P.~Bagrow}
\author[3]{Erik M.~Bollt}
\affil[1]{Department of Mathematics \& Statistics, University of Vermont, Burlington, VT, United States}
\affil[2]{Vermont Complex Systems Center, University of Vermont, Burlington, VT, United States}
\affil[3]{Department of Mathematics, Clarkson University, Potsdam, NY, United States}
\affil[*]{\corrauthinfo{james.bagrow@uvm.edu}{bagrow.com}}

\maketitle

\begin{abstract}
As network research becomes more sophisticated, it is more common than ever for researchers to find themselves not studying a single network but needing to analyze sets of networks.
An important task when working with sets of networks is network comparison, developing a similarity or distance measure between networks so that meaningful comparisons can be drawn. The best means to accomplish this task remains an open area of research.
Here we introduce a new measure to compare networks, the Network Portrait Divergence, that is mathematically principled, incorporates the topological characteristics of networks at all structural scales, and is general-purpose and applicable to all types of networks.
An important feature of our measure that enables many of its useful properties is that it is based on a graph invariant, the network portrait.
We test our measure on both synthetic graphs and real world networks taken from protein interaction data, neuroscience, and computational social science applications. 
The Network Portrait Divergence reveals important characteristics of multilayer and temporal networks extracted from data.
\end{abstract}


\bigskip
\noindent
\textbf{\textit{Code---}} \url{https://github.com/bagrow/network-portrait-divergence}

\section{Introduction}\label{sec:introduction}

Recent years have seen an explosion in the breadth and depth of network data across a variety of scientific domains~\cite{bader2004gaining,lazer2009life,Landhuis:2017aa}.
This scope of data challenges researchers, and new tools and techniques are necessary for evaluating and understanding networks.
It is now increasingly common to deal with multiple networks at once, from brain networks taken from multiple subjects or across longitudinal studies~\cite{whelan2012adolescent}, to multilayer networks extracted from high-throughput experiments~\cite{de2015structural}, to rapidly evolving social network data~\cite{palla2007quantifying,szell2010multirelational,bagrow2011collective}.
A common task researchers working in these areas will face is comparing networks, quantifying the similarities and differences between networks in a meaningful manner.
Applications for network comparison include comparing brain networks for different subjects, or the same subject before and after a treatment, studying the evolution of temporal networks~\cite{holme2012temporal}, classifying proteins and other sequences~\cite{shervashidze2009efficient,yanardag2015deep,niepert2016learning}, classifying online social networks~\cite{yanardag2015deep}, or evaluating the accuracy of statistical or generative network models~\cite{hunter2008goodness}. 
Combined with a clustering algorithm, a network comparison measure can be used to aggregate networks in a meaningful way, for coarse-graining data and revealing redundancy in multilayer networks~\cite{de2015structural}.
Treating a network comparison measure as an objective function, optimization methods can be used to fit network models to data.

Approaches to network comparison can be roughly divided into two groups, those that consider or require two graphs defined on the same set of nodes, and those that do not. 
The former eliminates the need to discover a mapping between node sets, making comparison somewhat easier.
Yet, two networks with identical topologies may have no nodes or edges in common simply because they are defined on different sets of nodes.
While there are scenarios where assuming the same node sets is appropriate---for example, 
when comparing the different layers of a multilayer network one wants to capture explicitly the correspondences of nodes between layers~\cite{de2015structural}---here we wish to relax this assumption and allow for comparison without node correspondence, where no nodes are necessarily shared between the networks.

A common approach for comparison without assuming node correspondence is to build a comparison measure using a \emph{graph invariant}.
Graph invariants are properties of a graph that hold for all isomorphs of the graph.
Using an invariant mitigates any concerns with the encoding or structural representation of the graphs, and the comparison measure is instead focused entirely on the topology of the network.
Graph invariants may be probability distributions.
Suppose $P$ and $Q$ represent two graph-invariant distributions corresponding to graphs $G_1$ and $G_2$, respectively.
Then, a common approach to comparing $G_1$ and $G_2$ is by comparing $P$ and $Q$.
Information theory provides tools for comparing distributions, such as the Jensen-Shannon divergence:
\begin{equation}
	\Djs(P, Q) = \frac{1}{2}\KL{P}{M} + \frac{1}{2}\KL{Q}{M}
\end{equation}
where $\KL{P}{Q}$ is the Kullback-Leibler (KL) divergence (or relative entropy) between $P$ and $Q$ and $M = (P + Q)/2$ is the mixture distribution of $P$ and $Q$. 
The Jensen-Shannon divergence has a number of nice properties, including that it is symmetric and normalized, 
making it a popular choice for applications such as ours~\cite{PhysRevX.6.041062,chen2018complex}.
In this work, we introduce a novel graph-invariant distribution that is general and free of assumptions and we can then use common information-theoretic divergences such as Jensen-Shannon to compare networks via these graph invariants.

The rest of this paper is organized as follows.
In Sec.~\ref{sec:networkportraits} we describe network portraits~\cite{bagrow2008portraits}, a graph invariant matrix representation of a network that is useful for visualization purposes but also capable of comparing pairs of networks.
Section~\ref{sec:infothrnetcomp} introduces Network Portrait Divergences, a principled information-theoretic measure for comparing networks, building graph-invariant distributions using the information contained within portraits. 
Network Portrait Divergence has a number of desirable properties for a network comparison measure.
In Sec.~\ref{sec:results} we apply this measure to both synthetic networks (random graph ensembles) and real-world datasets (multilayer biological and temporal social networks), demonstrating its effectiveness on practical problems of interest.
Lastly, we conclude in Sec.~\ref{sec:discussion} with a discussion of our results and future work.

\section{Network portraits}
\label{sec:networkportraits}

Network portraits were introduced in~\cite{bagrow2008portraits} as a way to visualize and encode many structural properties of a given network. Specifically, the network portrait $B$ is the array with $(\ell,k)$ elements
\begin{equation}
B_{\ell,k} \equiv \mbox{the number of nodes who have $k$ nodes at distance $\ell$}
\label{eqn:bmatrixdefinition}
\end{equation}
for $0 \leq \ell \leq d$ and $0\leq k \leq N-1$, where distance is taken as the shortest path length and $d$ is the graph's diameter
\footnote{Note that a distance $\ell=0$ is admissible, with two nodes $i$ and $j$ at distance 0 when $i=j$. This means that the matrix $B$ so defined has a zeroth row. It also has a zeroth column, as there may be nodes that have zero nodes at some distance $\ell$. This occurs for nodes with eccentricity less than the graph diameter.}.
The elements of this array are computed using, e.g.,  Breadth-First Search.
Crucially, no matter how a graph's nodes are ordered or labeled the portrait is identical.
We draw several example networks and their corresponding portraits in Fig.~\ref{fig:exampleportraits}.

\begin{figure*}[t!]
\centering
\includegraphics[width=\textwidth]{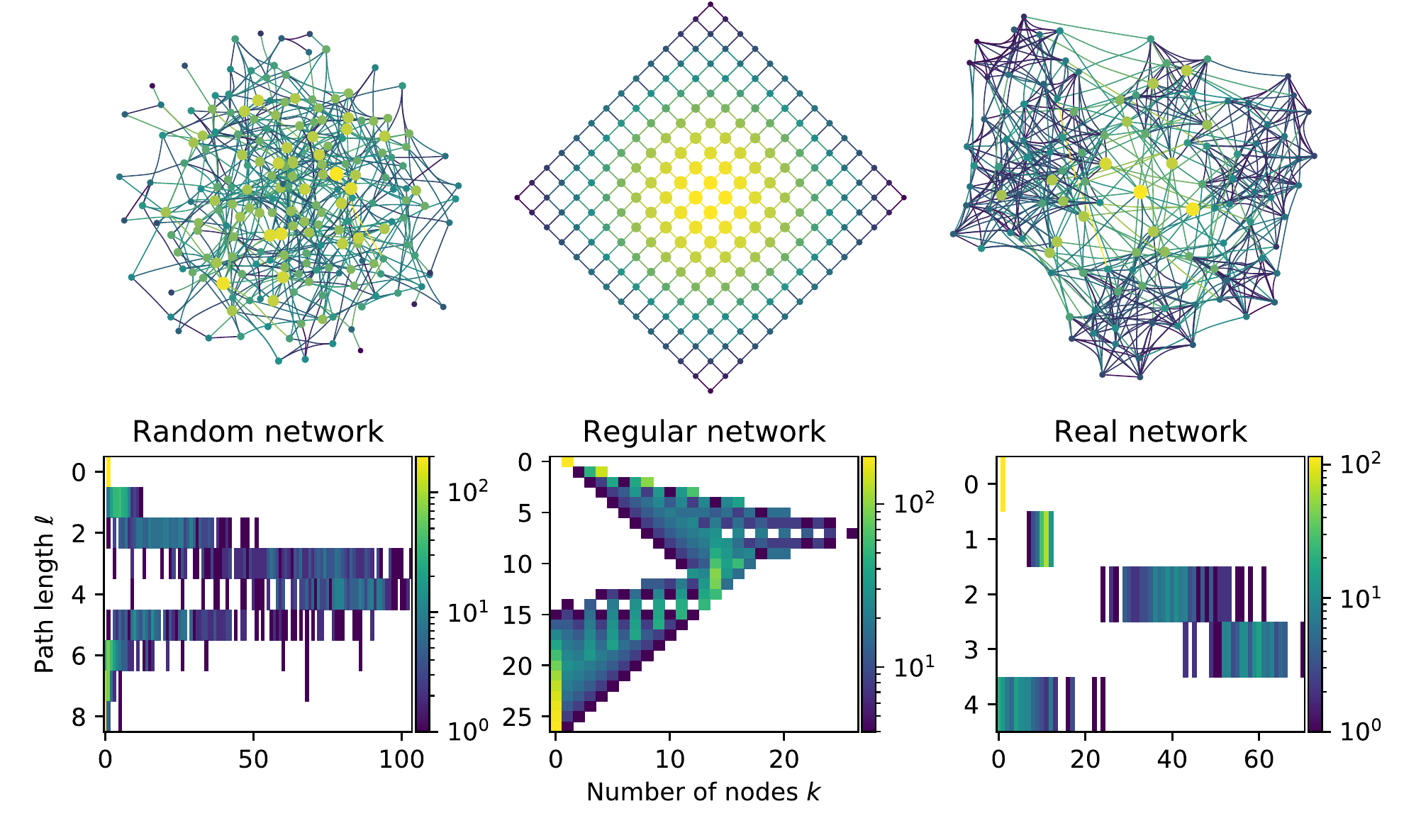} 
\caption{Example networks and their portraits. The random network is an \ER{} graph while the real network is the NCAA Division-I football network~\cite{park2005network}.
Colors denote the entries of the portrait matrix $B$ (Eq.~\eqref{eqn:bmatrixdefinition}; white indicates $B_{\ell,k}=0$).
\label{fig:exampleportraits}
}
\end{figure*}

This matrix encodes many structural features of the graph.
The zeroth row stores the number of nodes $N$ in the graph: $$B_{0,k} = N \delta_{k,1}.$$
The first row captures the degree distribution $P(k)$: $$B_{1,k} = N P(k),$$ as neighbors are at distance $\ell = 1$.
The second row captures the distribution of next-nearest neighbors, and so forth for higher rows.
The number of edges $M$ is $\sum_{k=0}^{N} k B_{1,k} = 2 M$.
The graph diameter $d$ is 
$$d = \max\{\ell \mid B_{\ell,k} > 0 \mbox{~for~} k > 0\}.$$
The shortest path distribution is also captured: the number of shortest paths of length $\ell$
is $\frac{1}{2}\sum_{k=0}^{N} k B_{\ell,k}$.
And the portraits of random graphs present very differently from highly ordered structures such as lattices (Fig.~\ref{fig:exampleportraits}), demonstrating how dimensionality and regularity of the network is captured in the portrait~\cite{bagrow2008portraits}.

One of the most important properties of portraits is that they are a graph invariant:

\begin{definition}
A \textbf{graph invariant} is a property of a graph that is invariant under graph isomorphism, i.e., it is a function $f$ such that $f(G) = f(H)$ whenever $G$ and $H$ are isomorphic graphs.
\end{definition}

\begin{theorem}
The network portrait (Eq.~\eqref{eqn:bmatrixdefinition}) is a graph invariant.
\end{theorem}

\begin{proof}
Let $f: V_G \to V_H$ be a vertex bijection between two graphs $G = (V_G, E_G)$ and $H = (V_H, E_H)$ such that the number of edges between every pair of vertices $(i,j)$ in $G$ equals the number of edges between their images $(f(i),f(j))$ in $H$. 
Then $G$ and $H$ are isomorphic. 
Let $\ell_G(i,j)$ be the length of the shortest path between nodes $i$ and $j$ in $G$. 
For two isomorphic graphs $G$ and $H$, $\ell_G(i,j) = \ell_H(f(i), f(j))$ for all $i$ and $j$ in $G$, since the shortest path tuples $(i, \ldots, j)$ in $G$ and $(f(i), \ldots, f(j))$ in $H$ are the same length. 
All elements in the matrix $B(G)$ are computed by aggregating the values of $\ell_G(i,j)$. 
Therefore, $B(G) = B(H)$.
\end{proof}

Note that the converse is not necessarily true: that $f(G) = f(H)$ does not imply that $G$ and $H$ are isomorphic. 
As a counter-example, the non-isomorphic distance-regular dodecahedral and Desargues graphs have equal portraits~\cite{bagrow2008portraits}.

\medskip{}\smallskip{}
\noindent\textbf{Portraits of weighted networks~~}
The original work defining network portraits \cite{bagrow2008portraits} did not consider weighted networks, where a scalar quantity $w_{ij}$ is associated with each $(i,j) \in E$.
An important consideration is that path lengths for weighted networks are generally computed by summing edge weights along a path, leading to path lengths $\ell \in \mathbb{R}$ (typically) instead of path lengths $\ell \in \mathbb{Z}$.
To address this, in the \suppmat{} (\SM{}) we generalize the portrait to weighted networks, specifically accounting for how real-valued path lengths must change the definition of the matrix $B$.

\subsection{Comparing networks by comparing portraits}

Given that a graph $G$ admits a unique $B$-matrix makes these portraits a valuable tool for network comparison. Instead of directly comparing graphs $G$ and $G^{\prime}$, we may compute their portraits $B$ and $B^{\prime}$, respectively, and then compare these matrices.
We review the comparison method in our previous work~\cite{bagrow2008portraits}.
First, compute for each portrait $B$ the matrix $C$ consisting of row-wise cumulative distributions of $B$:
\begin{equation}
C_{\ell,k} = \left. \sum_{j=0}^{k} B_{\ell,j} \middle/ \sum_{j=0}^{N} B_{\ell,j} \right..
\end{equation}
The row-wise Kolmogorov-Smirnov test statistic $K_\ell$ between corresponding rows $\ell$ in $C$ and $C^{\prime}$: $$K_\ell = \max \left| C_{\ell,k} - C_{\ell,k}^{\prime} \right|$$ allows a metric-like graph comparison.
This statistic defines a two-sample hypothesis test for whether or not the corresponding rows of the portraits are drawn from the same underlying, unspecified distribution.
If the two graphs have different diameters, the portrait for the smaller diameter graph can be expanded to the same size as the larger diameter graph by defining empty shells $\ell > d$ as $B_{\ell,k} = N \delta_{0,k}$.
Lastly, aggregate the test statistics for all pairs of rows using a weighted average to define the similarity $\Delta(G,G^{\prime})$ between $G$ and $G^{\prime}$:
\begin{equation}
\Delta(G,G^{\prime}) \equiv \Delta(B,B^{\prime}) = \frac{\sum_\ell \alpha_\ell K_\ell}{\sum_\ell \alpha_\ell},
\label{eqn:origComparison}
\end{equation}
where 
\begin{equation}
\alpha_\ell = \sum_{k > 0} B_{\ell,k} + \sum_{k > 0} B_{\ell,k}^{\prime}
\label{eqn:origComparisonWeights}
\end{equation}
is a weight chosen to increase the impact of the lower, more heavily occupied shells.

While we did develop a metric-like quantity for comparing graphs based on the KS-statistics (Eqs.~\eqref{eqn:origComparison} and \eqref{eqn:origComparisonWeights}), we did not emphasize the idea.
Instead, the main focus of the original portraits paper was on the use of the portrait for visualization. 
In particular, Eq.~\eqref{eqn:origComparison} is somewhat ad hoc. Here we now propose a stronger means of comparison using network portraits that is interpretable and grounded in information theory.

\section{An information-theoretic approach to network comparison}
\label{sec:infothrnetcomp}

Here we introduce a new way of comparing networks based on portraits. 
This measure is grounded in information theory, unlike the previous, ad hoc comparison measure, and has a number of other desirable attributes we discuss below.

The rows of $B$ may be interpreted as probability distributions:
\begin{equation}
P(k\mid\ell) = \frac{1}{N} B_{\ell,k}
\label{eqn:PrkGivenL}
\end{equation}
is the (empirical) probability that a randomly chosen node will have $k$ nodes at distance $\ell$.
This invites an immediate comparison \emph{per row} for two portraits:
\begin{equation}
\KL{P(k\mid\ell)}{Q(k\mid\ell)}= \sum_k P(k\mid\ell) \log \frac{P(k\mid\ell)}{Q(k\mid\ell)},
\label{eqn:KLperRow}
\end{equation}
where $\KL{p}{q}$ is the Kullback-Liebler (KL) divergence between two distributions $p$ and $q$, and $Q$ is defined as per Eq.~\eqref{eqn:PrkGivenL} for the second portrait (i.e., $Q(k\mid\ell) = \frac{1}{N^{\prime}} B_{\ell,k}^{\prime}$).
The KL-divergence admits an information-theoretic interpretation that describes how many extra bits are needed to encode values drawn from the distribution $P$ if we used the distribution $Q$ to develop the encoding instead of $P$.

However, while this seems like an appropriate starting point for defining a network comparison, Eq.~\eqref{eqn:KLperRow} has some drawbacks:
\begin{enumerate}\itemsep=0pt
\item $\KL{P(k)}{Q(k)}$ is undefined if there exists a value of $k$ such that $P(k)>0$ and $Q(k)=0$. Given that rows of the portraits are computed from individual networks, which may have small numbers of nodes, this is likely to happen often in practical use.

\item The KL-divergence is not symmetric and does not define a distance. 

\item Defining a divergence for each pair of rows of the two matrices gives $\max(d,d^{\prime})+1$ separate divergences, where $d$ and $d^{\prime}$ are the diameters of $G$ and $G^{\prime}$, respectively.
To define a scalar comparison value (a similarity or distance measure) requires an appropriate aggregation of these values, just like the original approach proposed in \cite{bagrow2008portraits}; we return to this point below.

\end{enumerate}
The first two drawbacks can be addressed by moving away from the KL-divergence and instead using, e.g., the Jensen-Shannon divergence or Hellinger distance. 
However, the last concern, aggregating over $\max(d,d^\prime)+1$ difference quantities, remains for those measures as well.

Given these concerns, we propose the following, utilizing the shortest path distribution encoded by the network portraits. 
Consider choosing two nodes uniformly at random with replacement. 
The probability that they are connected is
\begin{equation}
	\frac{\sum_c n_c^2}{N^2},
\end{equation}
where $n_c$ is the number of nodes within connected component $c$, the sum $\sum_c n_c^2$ runs over the number of connected components, and the $n_c$ satisfy $\sum_c n_c = N$.
Likewise, the probability the two nodes are 
at a distance $\ell$ from one another is
\begin{equation}
\frac{\text{\# paths of length }\ell}{\text{\# paths}} = \frac{1}{\left(\sum_c n_c^{2}\right)} \sum_{k=0}^{N}k B_{\ell,k}. 
\end{equation}
Lastly, the probability that one of the two nodes has $k-1$ other nodes at distance $\ell$ is given by
\begin{equation}
\frac{
k B_{\ell,k}
}{
	\sum_{k^{\prime}} k^{\prime} B_{\ell, k^{\prime}} }
.	
\end{equation}

We propose to combine these probabilities into a single distribution that encompasses the distances between nodes weighted by the ``masses'' or prevalences of other nodes at those distance, giving
us the
probability for choosing a pair of nodes at distance $\ell$ and for one of the two randomly chosen nodes to have $k$ nodes at that distance $\ell$:
\begin{equation}
P(k, \ell) = \frac{	\sum_c n_c^2	} {	N^2	}  \frac{	\sum_{k^{\prime}} k^{\prime} B_{\ell, k^{\prime}}	}{	\sum_{c} n_c^2	} \frac{	k B_{\ell,k}	}{	\sum_{k^{\prime}} k^{\prime} B_{\ell, k^{\prime}}	}  	= \frac{k B_{\ell,k}}{N^2}
\end{equation}
and likewise for $Q(k, \ell)$ using $B^{\prime}$ instead of $B$.
However, this distribution is not normalized
\begin{equation}
	\sum_k \sum_\ell k B_{\ell,k} = \sum_c n_c^2 \neq N^2\end{equation}
unless the graph $G$ is connected.
It will be advantageous for this distribution to be normalized in all instances, 
therefore, we condition this distribution on the two randomly chosen nodes being connected:
\begin{equation}
P(k,\ell) = 
\frac{\sum_{k^{\prime}} k^{\prime} B_{\ell, k^{\prime}}}{\sum_c n_c^2}
\frac{k B_{\ell,k}}{\sum_{k^{\prime}} k^{\prime} B_{\ell, k^{\prime}} }
= 	\frac{k B_{\ell,k} }{\sum_c n_c^2}.
\label{eqn:Pkell}
\end{equation}
This now defines a single (joint) distribution $P$ ($Q$) for all rows of $B$ ($B^\prime$) which can then be used to define a single KL-divergence between two portraits:
\begin{equation}
\KL{P(k,\ell)}{Q(k,\ell)} = \sum_{\ell=0}^{\max(d,d^{\prime})} \sum_{k=0}^{N} P(k,\ell) \log \frac{P(k,\ell)}{Q(k,\ell)}
\label{eqn:KLpq}
\end{equation}
where the log is base 2.

\begin{definition}
The \textbf{Network Portrait Divergence} $\Djs(G,\Gp)$ between two graphs $G$ and $\Gp$ is the Jensen-Shannon divergence as follows,
\begin{equation}
\Djs(G, G^{\prime}) \equiv \frac{1}{2}\KL{P}{M} + \frac{1}{2}\KL{Q}{M}
\end{equation}
where $M = \frac{1}{2}\left(P+Q\right)$ is the mixture distribution of $P$ and $Q$. 
Here $P$ and $Q$ are defined by Eq.~\eqref{eqn:Pkell}, and the $\KL{\cdot}{\cdot}$ is given by Eq.~\eqref{eqn:KLpq}.
\label{def:portraitdiv}
\end{definition}

The Network Portrait Divergence $0 \leq \Djs \leq 1$ provides a single value to quantify the dissimilarity of the two networks by means of their distance distributions, with smaller $\Djs$ for more similar networks and larger $\Djs$ for less similar networks. Unlike the KL divergence, $\Djs$ is symmetric, $\Djs(G,\Gp) = \Djs(\Gp,G)$ and $\sqrt{\Djs}$ is a metric~\cite{1207388}.

The Network Portrait Divergence has a number of desirable properties\footnote{Code implementing Network Portrait Divergence is available at \url{https://github.com/bagrow/network-portrait-divergence}.}.
It is grounded in information theory, which provides principled interpretation of the divergence measure.
It compares networks based entirely on the structure of their respective topologies:
the measure is independent of how the nodes in the network are indexed and, further, does not assume the networks are defined on the same set of nodes. 
Network Portrait Divergence is relatively computationally efficient; unlike graph edit distance measures, for example, because Network Portrait Divergence is based on a graph invariant and expensive optimizations such as ``node matching'' are not needed.
Both undirected and directed networks are treated naturally, and disconnected networks can be handled without any special problems.
Using the generalization of network portraits to weighted networks (see \SM{}), the Network Portrait Divergence can also be used to compare weighted networks.
Lastly, all scales of structure within the two networks contribute simultaneously to the Network Portrait Divergence via the joint neighbor-shortest path length distribution (Eq.~\eqref{eqn:Pkell}), from local structure to motifs to the large scale connectivity patterns of the two networks.

\section{Results}
\label{sec:results}

Now we explore the use of the Network Portrait Divergence (Definition~\ref{def:portraitdiv}) to compare networks across a variety of applications. 
We study both synthetic example graphs, to benchmark the behavior of the comparison measure under idealized conditions. Then real world network examples are presented to better capture the types of comparison tasks researchers may encounter.

\subsection{Synthetic networks}

To understand the performance of the Network Portrait Divergence, we begin here by examining how it relates different realizations of the following synthetic graphs:
\begin{enumerate}
\item \ER{} (ER) graphs $G(N,p)$~\cite{erdos1959random}, the random graph on $N$ nodes where each possible edge exists independently with constant probability $p$; 
\item \BA{} (BA) graphs $G(N,m)$~\cite{barabasi1999emergence}, where $N$ nodes are added sequentially to a seed graph and each new node attaches to $m$ existing nodes according to preferential attachment.
\end{enumerate}

Figure~\ref{fig:distrPortraitDivERBA} shows the distributions of $\Djs{}$ for different realizations of ER and BA graphs with the same parameters, as well as the cross-comparison distribution using $\Djs{}$ to compare one realization of an ER graph with one realization of a BA graph where both graphs have the same average degree $\avg{k}$ but other parameters may vary.
Overall we note that realizations drawn from the same ensemble have relatively small Network Portrait Divergence, with ER falling roughly in the range $0 < \Djs{} < 0.5$ and BA between $0.1 < \Djs{} < 0.4$ (both smaller than the $\max(\Djs) = 1$).
In contrast, $\Djs{}$ is far higher when comparing one ER to one BA graph, with $\Djs > 0.6$ in our simulations.
The Network Portrait Divergence captures the smaller differences between different realizations of a single networks ensemble with correspondingly smaller values of $\Djs{}$ while the larger differences between networks from different ensembles are captured with larger values of $\Djs{}$.

\begin{figure}[t!]
\centering
{\includegraphics[width=\factor\columnwidth]{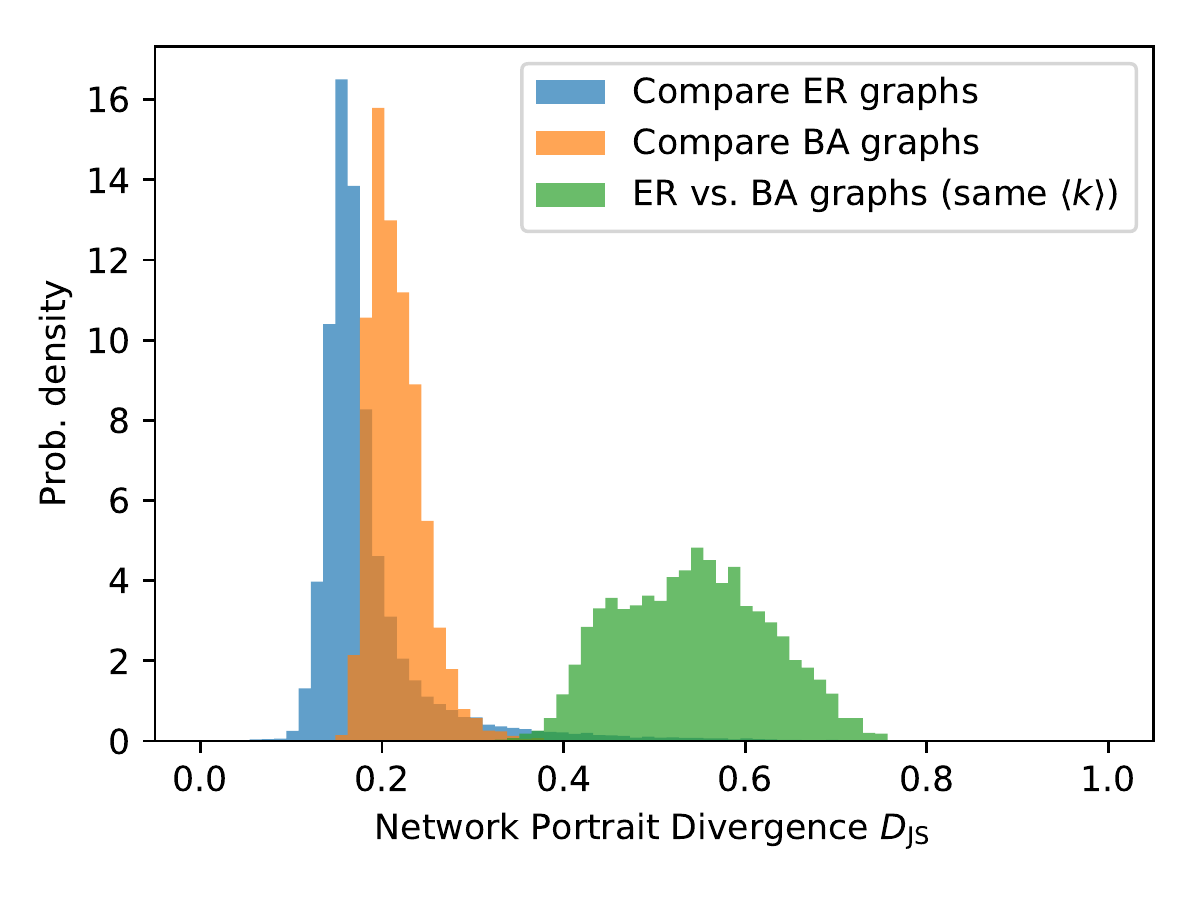}}
\caption{Comparing network models. 
As a simple starting point, here we compare pairs of \ER{} (ER) realizations with the same parameters, pairs of \BA{} (BA) realizations with the same parameters, and pairs of one ER graph versus one BA graph (but with the same number of nodes and average degree).
For these examples, Network Portrait Divergence has high discriminative power.
\label{fig:distrPortraitDivERBA}
}
\end{figure}

\subsection{Measuring network perturbations with Network Portrait Divergence}

Next, we ask how well the Network Portrait Divergence measures the effects of network perturbations. 
We performed two kinds of rewiring perturbations to the links of a given graph $G$: (i) \emph{random rewiring}, where each perturbation consists of deleting an existing link chosen uniformly at random and inserting a link between two nodes chosen uniformly at random; and (ii) \emph{degree-preserving rewirings}~\cite{li2017network}, where each perturbation consists of removing a randomly chosen pair of links $(i,j)$ and $(u,v)$ and inserting links $(i,u)$ and $(j,v)$. 
The links $(i,j)$ and $(u,v)$ are chosen such that $(i,u)\not\in E$ and $(j,v)\not\in E$, ensuring that the degrees of the nodes are constant under the rewiring.

We expect that random rewirings will lead to a stronger change in the network than the degree-preserving rewiring.
To test this, we generate an ER or BA graph $G$, apply a fixed number $n$ of rewirings to a copy of $G$, and use the Network Portrait Divergence to compare the networks before and after rewirings. 
Figure \ref{fig:rewiringsPortraitDivergenceERBA} shows how $\Djs$ changes on average as a function of the number of rewirings, for both types of rewirings and both ER and BA graphs.
The Network Portrait Divergence increases with $n$, as expected.
Interestingly, below  $n \approx 100$ rewirings, the different types of rewirings are indistinguishable, but for $n > 100$ we see that random rewirings lead to a larger divergence from the original graph than degree-preserving rewirings.
This is especially evident for BA graphs, where the scale-free degree distribution is more heavily impacted by the random rewiring than for ER graphs.
The overall $\Djs$ is also higher in value for BA graphs than ER graphs.
This is plausible because the ER graph is already maximally random, whereas many correlated structures exist in a random realization of the BA graph model that can be destroyed by perturbations ~\cite{albert2002statistical}.

\begin{figure*}[t!]
\centering
\includegraphics[width=\textwidth]{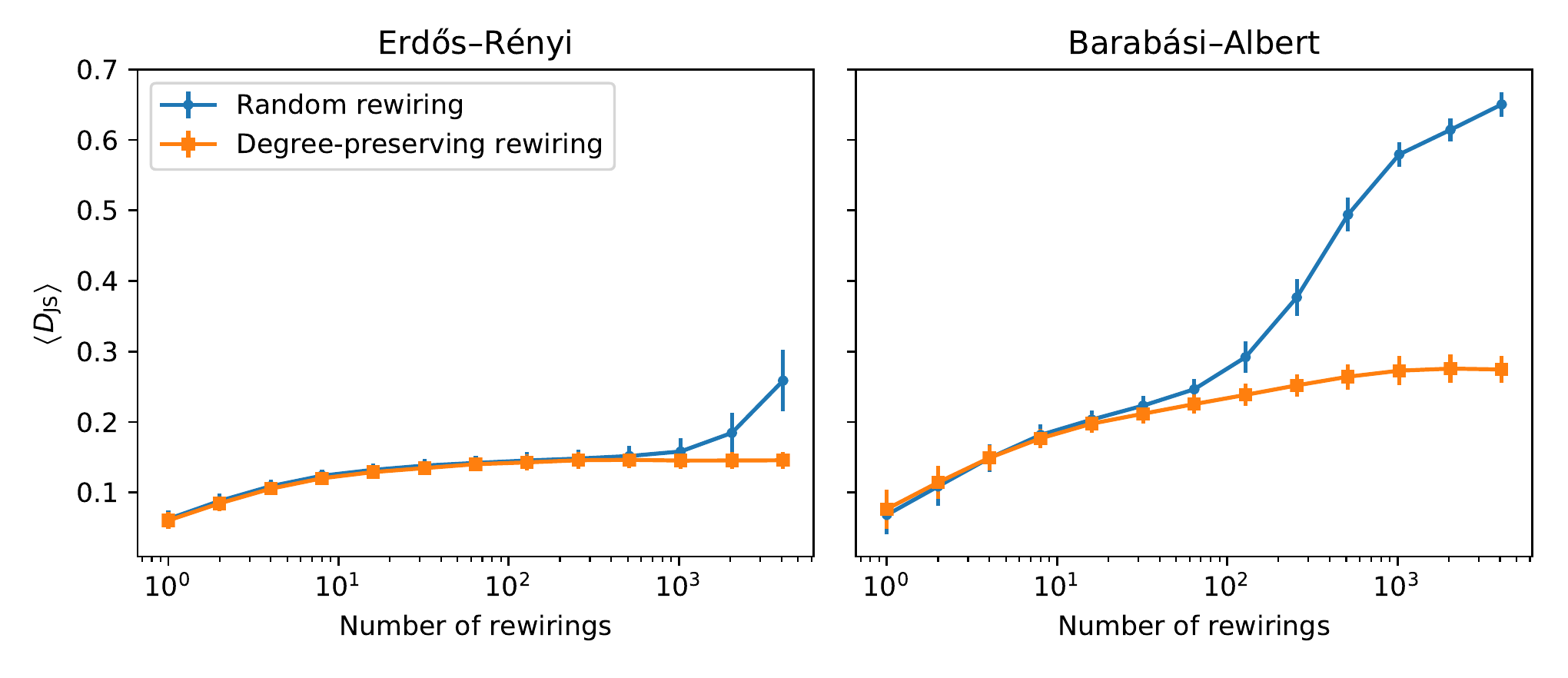}
\caption{Expected Network Portrait Divergence between an unmodified graph and a copy perturbed by edge rewiring. 
(left)  \ER{} graph ($N = 300, p = 3/299$);
(right) \BA{} graph ($N = 300, m = 3$).
A random rewiring is the deletion of an edge chosen uniformly at random followed by the insertion of a new edge between two nodes chosen uniformly at random.
Degree-preserving rewiring chooses a pair of edges $(u,v)$ and $(x,y)$ and rewires them across nodes to $(u,x)$ and $(v,y)$ such that the degrees of the chosen nodes remain unchanged~\cite{li2017network}.
Errorbars denote $\pm$ 1 s.d.
\label{fig:rewiringsPortraitDivergenceERBA}
}
\end{figure*}

\subsection{Comparing real networks}

We now apply the Network Portrait Divergence to real world networks, to evaluate its performance when used for several common network comparison tasks. 
Specifically, we study two real-world multiplex networks, using $\Djs$ to compare across the layers of these networks.
We also apply $\Djs$ to a temporal network, measuring how the network changes over time.
This last network has associated edge weights, and we consider it as both an unweighted and a weighted network.

The datasets for the three real-world networks we study are as follows:
\begin{description}
\item[\emph{Arabidopsis} GPI network]
The Genetic and Protein Interaction (GPI) network of \emph{Arabidopsis Thaliana} taken from BioGRID 3.2.108~\cite{stark2006biogrid,de2015structural}. 
This network consists of 6,980 nodes representing proteins and 18,654 links spread across seven multiplex layers. 
These layers represent different interaction modes from direct interaction of protein and physical associations of proteins within complexes to suppressive and synthetic genetic interactions.
Full details of the interaction layers are described in \cite{de2015structural}.

\item[\celegans{} connectome]

The network taken from the nervous system of the nematode \celegans{}. 
This multiplex network consists of 279 nodes representing non-pharyngeal neurons and 5,863 links representing electrical junctions and chemical synapses.
Links are spread across three layers: 
an electric junction layer (17.6\% of links), 
a monadic chemical synapse layer (28.0\% of links), 
and a polyadic chemical synapse layer (54.4\% of links)~\cite{chen2006wiring,de2015muxviz}.
The first layer represents electrical coupling via gap junctions between neurons, while links in the other layers represent neurons coupled by neurotransmitters.
\celegans{} has many advantages as a model organism in general~\cite{brenner1974genetics,white1986structure}, and its neuronal wiring diagram is completely mapped experimentally~\cite{white1986structure,hall2007c}, making its connectome an ideal test network dataset. 

\item[Open source developer collaboration network]
This network represents the software developers working on open source projects hosted by IBM on GitHub (\url{https://github.com/ibm}).
This network is temporal, evolving over the years 2013-2017, allowing us to compare its development across time. 
Aggregating all activity, this network consists of 679 nodes and 3,628 links.
Each node represents a developer who has contributed to the source code of the project, as extracted from the git metadata logs~\cite{bird2009promises,kalliamvakou2014promises,klug2016understanding}. 
Links occur between developers who have edited at least one source code file in common, a simple measure of collaboration. To study this network as a \emph{weighted network}, we associate with each link $(i,j)$ an integer weight $w_{ij}$ equal to the number of source files edited in common by developers $i$ and $j$. 

\end{description}

For these data, the Network Portrait Divergence reveals several interesting facets of the multilayer structure of the Arabidopsis network (Fig.~\ref{fig:multiplex_networks}). First, both the direct interaction and physical association layers are the most different, both from each other and from the other layers, with the synthetic interactions layer being most different from the direct and physical layers.
The remaining five layers show interesting relationships, with suppressive and synthetic genetic interactions being the most distinct pair ($\Djs = 0.553$) among these five layers.
The additive and suppressive layers were also distinct ($\Djs = 0.365$).
De Domenico \emph{et al.}\ observed a similar distinction between the suppressive and synthetic layers~\cite{de2015structural}.

\begin{figure*}[t!]
\centering
{\includegraphics[width=\textwidth]{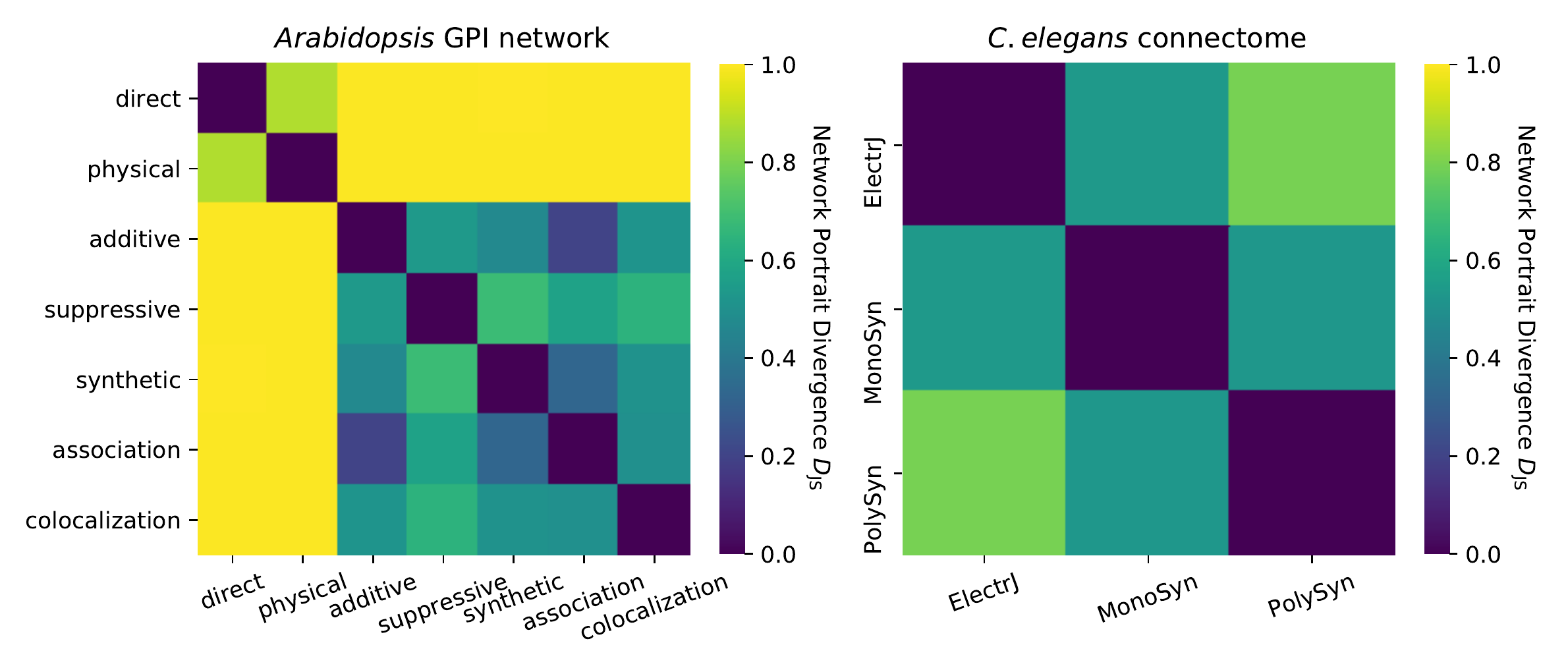}}
\caption{Using Network Portrait Divergence to compare across layers of multiplex networks.
The $(i,j)$ entry of each matrix illustrates the Network Portrait divergence between layer $i$ and layer $j$ of the corresponding multiplex network.
Network Portrait Divergence can reveal similarities and differences across the different multiplex layers.
\label{fig:multiplex_networks}
}
\end{figure*}

The multilayer \celegans{} network, consisting of only three layers, is easier to understand than \emph{Arabidopsis}. 
Here we find that the electrical junction layer is more closely related to the monadic synapse layer than it is to the polyadic synapse layer, while the polyadic layer is more closely related to the monadic synapse layer than to the electrical junction layer.
The \celegans{} data aggregated all polyadic synapses together into one layer accounting for over half of the total links in the network, but it would be especially interesting to determine what patterns for dyadic, triadic, etc.\ synapses can be revealed with the Network Portrait Divergence.

The third real-world network we investigate is a temporal network (Fig.~\ref{fig:ibmnetwork}).
This network encodes the collaboration activities between software developers who have contributed to open source projects owned by IBM on GitHub.com. Here each node represents a developer and a links exist between two developers when they have both edited at least one file in common among the source code hosted on GitHub.
This network is growing over time as more projects are open-sourced by IBM, and more developers join and contribute to those projects.
We draw the IBM developer network for each year from 2013 to 2017 in Fig.~\ref{fig:ibmnetwork}A, while Fig.~\ref{fig:ibmnetwork}B shows the change in size of these networks over time.
Lastly, Fig.~\ref{fig:ibmnetwork}C demonstrates how $\Djs$ captures patterns in the temporal evolution of the network, in particular revealing the structural similarity between 2015 and 2016, a period that showed a distinct slowdown in network growth compared with prior and subsequent years.

\begin{figure*}[t!]
\centering
{\includegraphics[width=\textwidth]{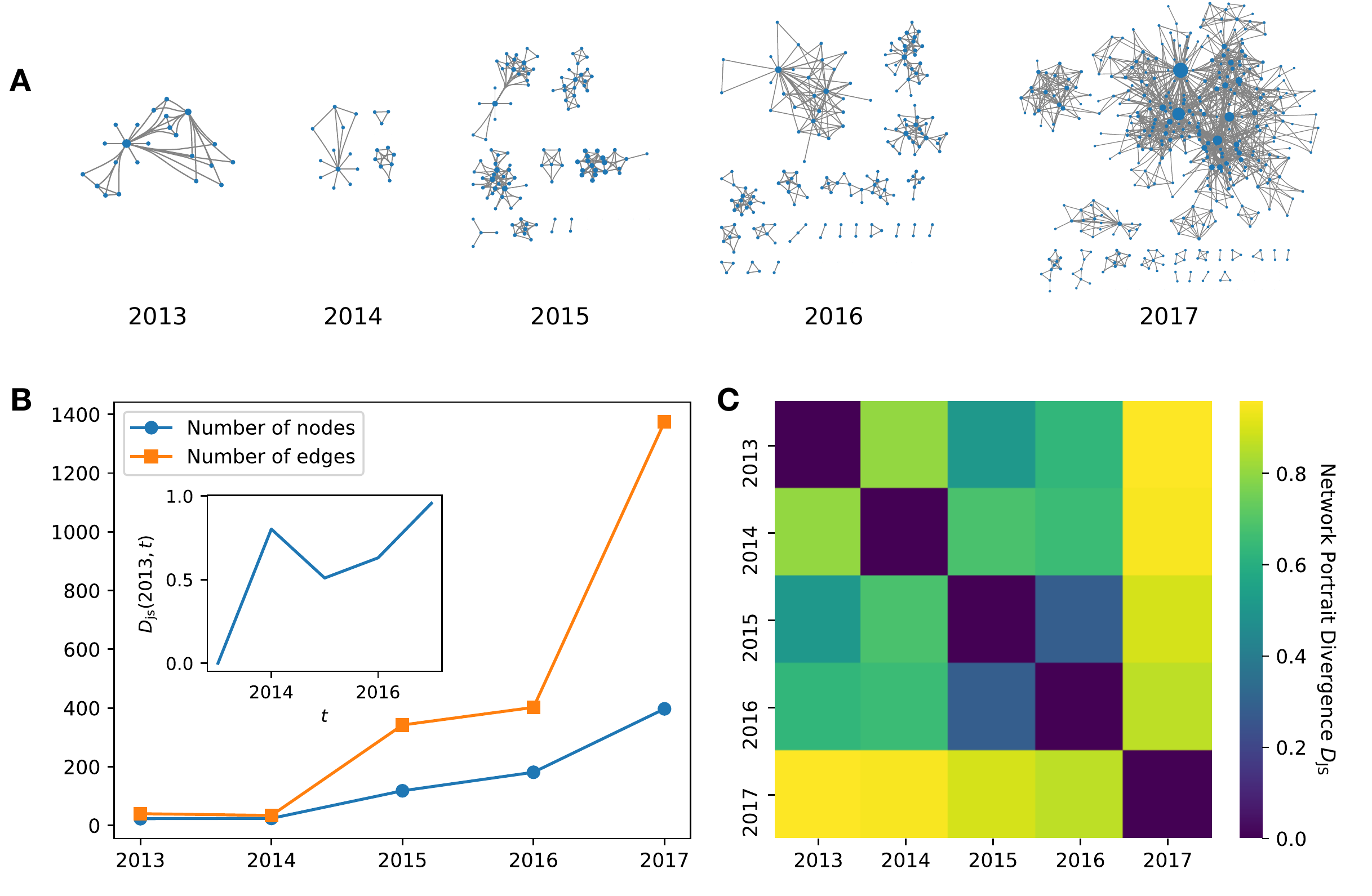}}
\caption{Network Portrait Divergence and the temporal evolution of the IBM GitHub collaboration network.
\lett{A} The IBM collaboration networks, each one aggregated from a year of developer activity on IBM's GitHub projects.
\lett{B} Growth of the IBM network over time.
The inset tracks the Network Portrait Divergence of the network away from the first year.
\lett{C} Network Portrait Divergences comparing networks across years. 
The decreased growth rate from 2015 to 2016 is captured by a very close divergence. 
Otherwise, the divergences are quite high, especially between the extreme years, demonstrating the dynamic changes the IBM collaboration underwent during the data window.
Panel C treats the network as unweighted. 
\label{fig:ibmnetwork}
}
\end{figure*}

Figure \ref{fig:weightedJSDgithub} shows the Network Portrait Divergence comparing different years of the IBM developer network, as per Fig.~\ref{fig:ibmnetwork} but now accounting for edge weights. 
Shortest path lengths were found using Dijkstra's algorithm based on reciprocal edge weights (see \SM{} for details).
The weighted portraits require binning the path length distributions (see \SM{}).
Here we show four such binnings, based on quantiles of the weighted path length distributions, from $b = 100$ bins (each bin accounts for 1\% of shortest paths) to $b = 10$ (each bin accounts for 10\% of shortest paths).
Note that each $\Djs$ is defined with its own portrait binning, as the quantiles of $\mathcal{L} = \mathcal{L}(G) \cup \mathcal{L}(\Gp)$ may vary across different pairs of networks, where
$\mathcal{L}(G) = \{\ell_{ij} \mid i,j \in V \land \ell_{ij} < \infty\}$ is the set of all unique shortest path lengths in $G$.
Overall, the relationships between the different time periods of the network do not depend strongly on the choice of binning, and we capture patterns across time similar to, though not identical to,  the patterns found analyzing the unweighted networks (shown in Fig.~\ref{fig:ibmnetwork}).

\begin{figure*}[t!]
\centering
{\includegraphics[width=0.975\textwidth]{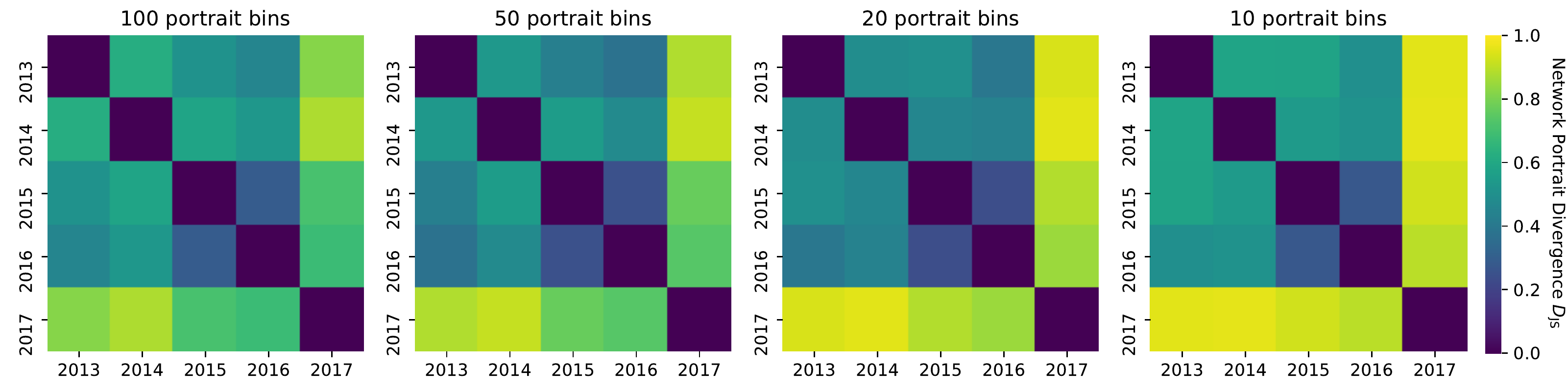}} 
\caption{Weighted Network Portrait Divergence for comparing the IBM developer network over time.
The weighted portraits depend on a choice of binning (see \SM{}); here we see that the
pattern of similarities between networks is robust to changes in binning.
\label{fig:weightedJSDgithub}
}
\end{figure*}

\section{Conclusion}
\label{sec:discussion}

In this paper we have introduced a measure, the Network Portrait Divergence, for comparing networks, and validated its performance on both synthetic and real-world network data. 
Network Portrait Divergence provides an information-theoretic interpretation that naturally encompasses all scales of structure within networks. 
It does not require the networks to be connected, nor does it make any assumptions as to how the two networks being compared are related, or indexed, or even that their node sets are equal.
Further, Network Portrait Divergence can naturally handle both undirected and directed, unweighted networks, and we have introduced a generalization for weighted networks.
The Network Portrait Divergence is based on a graph invariant, the network portrait. 
Comparison measures based on graph invariants are desirable as they will only be affected by the topology of the networks being studied, and not other externalities such as the format or order in which the networks are recorded or analyzed.
The computational complexity of the Network Portrait Divergence compares favorably to many other graph comparison measures, particularly spectral measures, but it remains a computation that is quadratic in the number of nodes of the graph. To scale to very large networks will likely require further efficiency gains, probably from approximation strategies to efficiently infer the shortest path distributions~\cite{potamias2009fast}.

Our approach bears superficial similarities with other methods.
Graph distances and shortest path length distributions are components common to many network comparison methods, including our own, although the Network Portrait Divergences utilizes a unique composition of all the shortest path length distributions for the networks being compared.
At the same time, other methods, including ours, use the Jensen-Shannon divergence to build comparison measures.
For example, the recent work of Chen et al.~\cite{chen2018complex} uses the Shannon entropy and Jensen-Shannon Divergence of a probability distribution computed by normalizing $e^{A}$, the exponential of the adjacency matrix also known as the communicability matrix.
This is an interesting approach, as are other approaches that examine powers of the adjacency matrix, but it suffers from a drawback: when comparing networks of different sizes, the underlying probability distributions must be modified in an ad hoc manner~\cite{chen2018complex}.
The Network Portrait Divergence, in contrast, does not need such modification.

The Network Portrait Divergence, and other methods, is based upon the Jensen-Shannon divergence between graph-invariant probability distributions, but many other information-theoretic tools exist for comparing distributions, including $f$-divergences such as the Hellinger distance or total variation distance, Bhattacharyya distance, and more. 
Using different measures for comparison may yield different interpretations and insights, and thus it is fruitful to better understand their use across different network comparison problems.

Network Portrait divergence lends itself well to developing statistical procedures when combined with suitable graph null models. 
For example, one could quantify the randomness of a structural property of a network by comparing the real network to a random model that controls for that property.
Further, to estimate the strength or significance of an observed divergence between graphs $G_1$ and $G_2$, one could generate a large number of random graph null proxies for either $G_1$ or $G_2$ (or both) and compare the divergences found between those nulls with the divergence between the original graphs. 
These comparisons could be performed using a $Z$-test or other statistical procedure, as appropriate. 
Exploring this and other avenues for Network Portrait Divergence-based statistical procedures is a fruitful direction for future work.

In general, because of the different potential interpretations and insights that researchers performing network comparison can focus on, the network comparison problem lacks quantitative benchmarks. 
These benchmarks are useful for comparing different approaches systematically.
However, the comparison problem is not as narrowly defined as, for example, graph partitioning, and thus effective methods may highlight very different facets of comparison.
While specific benchmarks can be introduced for specific facets, due to a lack of standardized, systematic benchmarks, most researchers introducing new comparison measures focus on a few tasks of interest, as we do here.
Bringing clarity to the literature by defining and validating an appropriate benchmarking suite would be a valuable contribution to the network comparison problem, but we consider this to be beyond the scope of our current work.
Instead, we have focused our method on highlighting several areas, particularly within real world applications but also using some intuitive synthetic scenarios, where the method is effective.

As network datasets increase in scope, network comparison becomes an increasingly common and important component of network data analysis. 
Measures such as the Network Portrait Divergence have the potential to help researchers better understand and explore the growing sets of networks within their possession.

\section*{Abbreviations}
JSD, Jensen-Shannon Divergence;
BA, Barab\'asi-Albert;
ER, Erd{\H o}s-R\'enyi;
KL, Kullback-Leibler;
GPI, Genetic and Protein Interaction.

\appendix

\section{Portraits and Network Portrait Divergences for weighted networks}
\label{sec:weighted}

The portrait matrix $B$ (Eq.~\eqref{eqn:bmatrixdefinition}) is most naturally defined for unweighted networks since the path lengths for unweighted networks count the number of edges traversed along the path to get from one node to another.
Since the number of edges is always integer-valued, these lengths can be used to define the rows of $B$. 
For weighted networks, on the other hand, path lengths are generally computed by summing edge weights along a path and will generally be continuous rather than integer-valued.

To generalize the portrait to weighted networks requires (i) using an algorithm for finding shortest paths accounting for edge weights (here we will use Dijkstra's algorithm~\cite{dijkstra1959note}), and (ii)
defining an appropriate aggregation strategy to group shortest paths by length to form the rows of $B$.
The algorithm for finding shortest paths defines the complexity of computing the portrait:
The single-source Dijkstra's algorithm with a Fibonacci heap runs in $\mathcal{O}(M + N \log N)$ time~\cite{fredman1987fibonacci} for a graph of $|V| = N$ nodes and $|E| = M$ edges.
This is more costly than the single-source Breadth-First Search algorithm we use for unweighted graphs, which runs in $\mathcal{O}(M + N)$ time. 
Computing $B$ requires all pairs of shortest paths, therefore the total complexity for computing a weighted portrait is $\mathcal{O}(M N + N^{2} \log N)$.
This again is more costly than the total complexity for the unweighted portrait, $\mathcal{O}(M N + N^{2})$ , but this is unavoidable as finding minimum-cost paths is generically more computationally intensive than finding minimum-length paths.

The simplest choice for aggregating shortest paths by length is to introduce a binning strategy for the continuous path lengths. 
Let $d_0=0 < d_1 < \cdots < d_{b+1}=L_\mathrm{max}$ define a set of $b$ intervals or bins, where $L_\mathrm{max}$ is the length of the longest shortest path.
Then the weighted portrait $B$ can be defined such that $B_{i,k} \equiv$ the number of nodes with $k$ nodes at distances $d_i \leq \ell < d_{i+1}$.
That is, the $i$-th row of the weighted portrait accounts for all shortest paths with lengths falling inside the $i$-th bin $[d_i, d_{i+1})$. (We also take the last bin to be inclusive on both sides, $[d_b, L_\mathrm{max}]$).

To compute $B$ using a binning requires determining the $b+1$ bin edges.
Here we consider a simple, adaptive binning based on quantiles of the shortest path distribution, but a researcher is free to adopt a different binning strategy as needed. 
Let $\mathcal{L}(G) = \{\ell_{ij} \mid i,j \in V \land \ell_{ij} < \infty\}$ be the set of all unique shortest path lengths between connected pairs of nodes in graph $G$.
We then define our binning to be the $b$ contiguous intervals that partition $\mathcal{L}$ into subsets of (approximately) equal size. Taking $b=100$, for example, ensures that each bin contains approximately 1\% of the shortest path lengths.
The number of bins $b$ can be chosen by the researcher to suit her needs, or automatically using any of a number of histogram binning rules such as Freedman-Diaconis~\cite{freedman1981histogram} or Sturges' Rule~\cite{sturges1926choice}.

Figure \ref{fig:weightedPortraits} shows the portrait for a weighted network, in this case taken from the IBM developer collaboration network.
Edge ($i,j$) in this network has associated non-negative edge weight $w_{ij} =$ the number of files edited in common by developers $i$ and $j$. 
The network is the union of the networks shown in Fig.~\ref{fig:ibmnetwork}A; we draw the giant connected component of this network in Fig.~\ref{fig:weightedPortraits}A. %
For this network, we consider shortest paths found using Dijkstra's algorithm with reciprocal edge weights, i.e., the ``length'' of a path $(i =i_0, i_1, i_2, \ldots, i_{n+1} = j)$ is $\ell_{ij} = \sum_{t=0}^{n} w_{i_t,i_{t+1}}^{-1},$ as larger edge weights define more closely related developers. However, this choice is not necessary in general.
The cumulative distribution of shortest path lengths, which we computed on all components of the network, is shown in Fig.~\ref{fig:weightedPortraits}B. 
Lastly, Fig.~\ref{fig:weightedPortraits}C shows the portrait $B$ for this network. 
For illustration, we draw the vertical positions of the rows in this matrix using the bin edges. These bin edges are highlighted on the cumulative distribution shown in Fig.~\ref{fig:weightedPortraits}B.

\begin{figure*}[t!]
\centering
\includegraphics[width=0.75\textwidth]{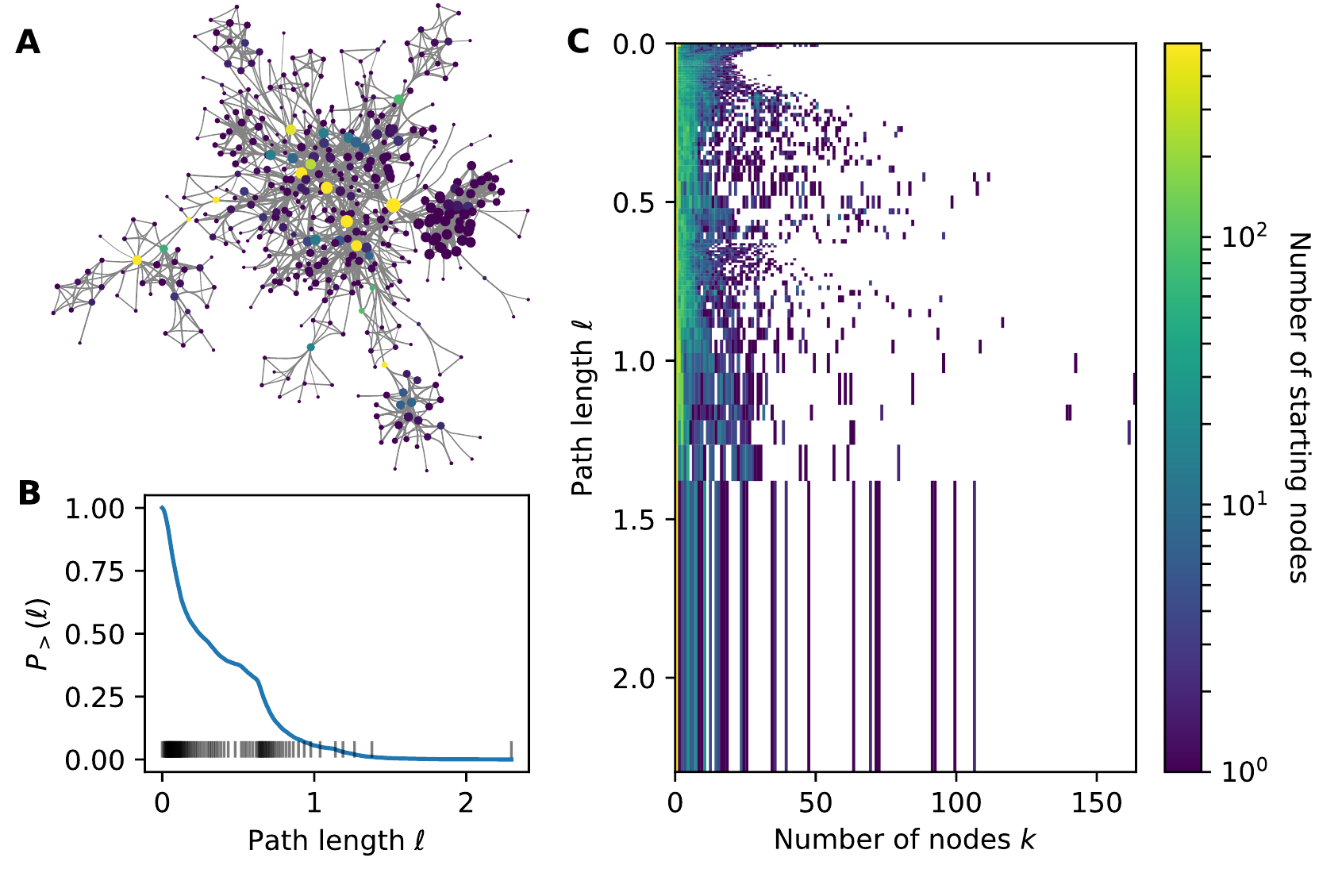}
\caption{A weighted network and its portrait.
\lett{A} The network of developers contributing to IBM projects on GitHub. 
This is the giant connected component of the union of all graphs shown in Fig.~\ref{fig:ibmnetwork}A.
The weight on edge $(i,j)$ represents the number of source files edited in common by developers $i$ and $j$.
Node size is proportional to degree; node color is proportional to betweenness centrality.
\lett{B} The cumulative distribution of shortest path lengths $\ell$ computed using Dijkstra's algorithm with reciprocal edge weights.
\lett{C} The weighted network portrait. The vertical marks in panel B denote the path length binning used in C.
\label{fig:weightedPortraits}
}
\end{figure*}

With a new definition for $B$ now in place for weighted networks, the Network Portrait Divergence can be computed exactly as before (Definition~\ref{def:portraitdiv}).
However, to compare portraits for two graphs $G$ and $\Gp$, it is important for the path length binning to be the same for both. We do this here by computing $b$ bins as quantiles of $\mathcal{L} = \mathcal{L}(G) \cup \mathcal{L}(\Gp)$ and then compute $B(G)$ and $B(\Gp)$ as before.
This ensures the rows of $B$ and $B^{\prime}$ are compatible in the distributions used within Definition~\ref{def:portraitdiv}.

\section*{Acknowledgments}
We thank Laurent H\'ebert-Dufresne for helpful comments.
J.P.B.\  was supported by the National Science Foundation under Grant No.\ IIS-1447634.  
E.M.B.\ would like to thank the Army Research Office (N68164-EG) and the Office of Naval Research (N00014-15-1-2093), and also DARPA.

\begin{singlespacing}

\end{singlespacing}

\end{document}